\documentclass[journal]{IEEEtran}
\usepackage{bm}
\usepackage{multicol}
\usepackage{mathtools}
\usepackage{amssymb}
\usepackage{amsmath}
\usepackage{color}

\usepackage{subfigure}
\usepackage{array}  
\usepackage{url}

\usepackage{amsthm}
\theoremstyle{plain}
\newtheorem{thm}{Theorem}

\newtheorem{cor}{Corollary}
\theoremstyle{definition}

\theoremstyle{remark}

\usepackage{algorithm}
\usepackage{algorithmic}

\usepackage[dvips]{graphicx}

\usepackage{cite}

\begin{document}
\title{Effective Capacity of Licensed-Assisted Access in Unlicensed Spectrum for 5G: From Theory to Application}

\author{Qimei Cui,~\IEEEmembership{Senior Member,~IEEE,}
        Yu Gu,~\IEEEmembership{Student Member,~IEEE,}
        Wei Ni,~\IEEEmembership{Senior Member,~IEEE,}\\
        and Ren Ping Liu,~\IEEEmembership{Senior Member,~IEEE}
\thanks{Manuscript received December 15, 2016; accepted March 6, 2017. }
\thanks{The work was supported by National Nature Science Foundation of China Project (Grant No. 61471058), International Cooperation NSFC Program (61461136002)£¬the National High Technology Program of China (2014AA01A701), Hong Kong, Macao and Taiwan Science and Technology Cooperation Projects (2016YFE0122900) and the 111 Project of China (B16006).}
\thanks{Q. Cui and Y. Gu are with the National Engineering Laboratory for Mobile Network Technologies, Beijing University of Posts and Telecommunications, Beijing 100876, China (e-mail: cuiqimei@bupt.edu.cn; guyu@bupt.edu.cn).}
\thanks{W. Ni is with the Digital Productivity and Services Flagship, Commonwealth Scientific and Industrial Research Organization (CSIRO), Sydney, N.S.W. 2122, Australia (e-mail: wei.ni@csiro.au).}
\thanks{R. P. Liu is with the Global Big Data Technologies Centre, University of Technology Sydney, Sydney, NSW 2007, Australia (e-mail: renping.liu@uts.edu.au).}}

\maketitle

\begin{abstract}
License-assisted access (LAA) is a promising technology to offload dramatically increasing cellular traffic to unlicensed bands. Challenges arise from the provision of quality-of-service (QoS) and the quantification of capacity, due to the distributed and heterogeneous nature of LAA and legacy systems (such as WiFi) coexisting in the bands. In this paper, we develop new theories of the effective capacity to measure LAA under statistical QoS requirements. A new four-state semi-Markovian model is developed to capture transmission collisions, random backoffs, and lossy wireless channels of LAA in distributed heterogeneous network environments. A closed-form expression for the effective capacity is derived to comprehensively analyze LAA. The four-state model is further abstracted to an insightful two-state equivalent which reveals the concavity of the effective capacity in terms of transmit rate. Validated by simulations, the concavity is exploited to maximize the effective capacity and effective energy efficiency of LAA, and provide significant improvements of 62.7\%
and 171.4\%, respectively, over existing approaches. Our results are of practical value to holistic designs and deployments of LAA systems.
\end{abstract}

\section{Introduction}
\IEEEPARstart{T}{he} past decade has witnessed the explosive growth of mobile traffic stemming from the prevalence of smart handset devices \cite{2}. It is predicted that the mobile traffic will grow astoundingly 1000-fold by 2020 \cite{cisco2014global,xin2}. The scarcity of spectrum becomes the bottleneck of this growth in fifth-generation (5G) networks, and one of the solutions is widely believed to be the unlicensed spectrum \cite{7,xin1}. Recently, the Third Generation Partnership Project (3GPP) standardization group has specified license-assisted access (LAA) to the unlicensed band, coexisting with legacy systems such as IEEE 802.11 WiFi \cite{3}. The design goal is to comply with any regional regulatory requirements, while achieving effective and fair coexistence between LAA and WiFi networks. Contention-based access techniques, exploiting Listen-before-talk (LBT), have been specified to alleviate the intrusion of LAA to WiFi \cite{3}. Although its early versions have already been standardized in 3GPP Release 13 for Long Term Evolution (LTE), LAA will remain a key topic of 5G. As a matter of fact, 3GPP Release 15 has itemized ``new radio (NR) based unlicensed access'' and ``Enhancements to LTE operation in unlicensed spectrum'', where evolutions of LAA will be standardized to allow 5G to access the unlicensed spectrum \cite{NR,NR1}.

A prominent challenge arising is to provide quality-of-service (QoS) to LAA in 5G distributed heterogeneous network environments \cite{3}. (In contrast, WiFi does not ensure precise QoS~\cite{Zhu2004}. The latest versions of WiFi, such as EDCA, claimed to have incorporated QoS, essentially provide relative priorities, and cannot guarantee QoS). This is because there is typically no policy to regulate the deployment of wireless transmitters in the unlicensed band. LAA base stations (LAA-BSs) can be deployed in an ad-hoc fashion \cite{3}. LAA-BSs need to contend with the WiFi systems and other randomly deployed LAA-BSs for transmissions. The delay of LAA traffic could be prolonged and the minimum data rate could be violated, both due to repeated collisions and subsequent retransmissions. The delay and the minimum data rate would also deteriorate as the nodes increase, due to intensifying transmission collisions. The distributed nature, ad-hoc deployment and stringent QoS requirements also pose a challenge to the comprehensive analysis of LAA. The analysis is important to quantify the capacity of LAA. It is of practical value to design and optimize LAA system parameters, e.g., transmit power and contention window (CW).

With the prevalence of WiFi in the unlicensed band, the coexistence between LAA and WiFi is a prominent issue to be addressed. A lot of studies have been conducted on the coexistence of LAA and WiFi in the unlicensed band. Earlier designs, such as almost blank sub-frame (ABS) \cite{12}, duty cycle \cite{9}, and interference avoidance \cite{Zhang2015}, were rigid and intrusive to WiFi. Exploiting LBT, recent designs have substantially reduced the intrusions \cite{10,18,16}. Some of the designs display strong resemblance and behave friendly to WiFi \cite{3,Voicu2016,15,17, Wang2016,Han2016,Yin2016}. However, the modeling and analysis of these WiFi-friendly designs have been to date focused on throughput with little consideration on QoS. On the other hand, effective capacity, quantifying the maximum arrival rate at the input of a First-In-First-Out (FIFO) buffer while guaranteeing the QoS at the output, has been developed to measure loss-less queueing systems \cite{20}. This measure has been recently extended to single wireless point-to-point links \cite{20}, \cite{21}, centralized wireless networks \cite{jin2015resource} and simplified WiFi networks with error-free wireless channels \cite{28}. However, these are inapplicable to LAA which is part of a distributed heterogeneous network with lossy wireless channels. To the best of our knowledge, the effective capacity is yet to be established for LAA.

In this paper, we establish a new theoretical framework to quantify the effective capacity of LAA under statistical QoS constraints. A new four-state semi-Markovian model is proposed to precisely capture transmission collisions, random backoffs, and lossy wireless channels of LAA in distributed heterogeneous network environments. A closed-form expression is derived to quantify the effective capacity of a LAA user equipment (LAA-UE) against its QoS requirements, instantaneous transmit rate, and the numbers of LAA-BSs and WiFi devices. Further, we prove the four-state model is equivalent to an abstract two-state semi-Markovian model which, in turn, reveals the concavity of the effective capacity in terms of transmit rate. By exploiting the concavity, the effective capacity and the effective energy efficiency can be maximized, demonstrating the value of the new theoretical framework to practical designs and deployments of LAA systems. Corroborated by simulations, our framework is able to accurately measure LAA systems, and also substantially improve the effective capacity and effective energy efficiency of the systems.

Our key contributions can be summarized as follows:
\begin{enumerate}
  \item New closed-form expressions to evaluate the effective capacity of LAA against the QoS, instantaneous transmit rate, and the number of WiFi and LAA devices is theoretically derived by developing a new four-state semi-Markovian model, which captures transmission collisions, random backoff and lossy wireless channels in distributed heterogeneous networks.
  \item The concavity of the effective capacity of LAA is revealed and proved.
  \item The concavity is exploited to maximize the effective capacity and the effective energy efficiency, providing significant improvements of 62.7\%
and 171.4\% over the existing approaches, respectively. The results are of practical value to holistic designs and deployments of LAA systems.
\end{enumerate}

The rest of this paper is organized as follows. In Section II, related works are reviewed. In Section III, the system model is described. In Section IV, we establish the theoretical framework to analyze the effective capacity of LAA and uncover its concavity, followed by the applications of the framework to the designs of LAA systems in Section V. Numerical results are provided in Section VI, followed by conclusions in Section VII.

\section{Related Works}

The coexistence of LAA and WiFi in the unlicensed band has recently drawn extensive attention. Earlier LAA designs, such as ABS \cite{12}, duty cycle \cite{9}, and interference avoidance \cite{Zhang2015}, were rigid and intrusive to WiFi. Furthermore, designs of ABS and transmit power was studied to improve the robustness of WiFi to LAA in the unlicensed band in \cite{13}. In \cite{11}, Q-learning was employed to dynamically configure the duty cycle of LAA transmissions, adapting to the density of WiFi devices; while in \cite{19}, the energy efficiency was maximized under the duty cycle. These works may alleviate the intrusion of LAA to WiFi to some extent.

Exploiting LBT, recent LAA approaches have become more friendly to WiFi, where a LAA device senses the unlicensed band before transmission. Some of the approaches require the LAA devices to transmit immediately after the band is sensed free \cite{10,18,16}. For example, routing and resource allocation were jointly designed to support non-real-time and real-time applications in the downlink of a cloud radio access network \cite{18}. The uplink was also studied, yet under the assumption of a simplified on/off WiFi interference model \cite{16}. These approaches are still intrusive in the sense that the LAA devices are given priority over WiFi devices for every transmission opportunity.

Other LBT based approaches allow the LAA devices to randomly delay (or back off) transmissions, even if the unlicensed band is sensed free \cite{3,Voicu2016,15,17, Wang2016,Han2016,Yin2016}. Resembling to WiFi, these approaches enable WiFi and LAA devices to contend in a fair fashion. To this end, they are less intrusive and WiFi-friendly. In \cite{Voicu2016}, a comparison study was conducted between the approaches using duty cycle and LBT, and revealed the effectiveness of LBT in the case of strong interference. In \cite{15}, such an approach was modeled as a Markov chain, and the throughput was evaluated. Stochastic geometry was applied to analyze the medium access probability of a LAA-BS in \cite{17}, and the asymptotic coverage probability and throughput of WiFi and LAA networks in \cite{Wang2016} under a simplified Carrier Sense Multiple Access with Collision Detection (CSMA/CA) model without exponential backoffs or the dynamics of the timer history. In \cite{Han2016}, a LBT-based MAC protocol was developed, where the transmission durations of LAA devices were optimized to maximize the system throughput. In \cite{Yin2016}, the CW size was designed to be adjustable, adapting to the rate requirements of LAA-UEs and the collision probability. However, none of these works have taken QoS, particularly, delay, into account.

In a different yet relevant context, the effective capacity was developed to measure queueing systems, where QoS is characterized statistically \cite{20}. The effective capacity was extended to a single collision-free point-to-point wireless link \cite{20}, \cite{21}, and a collision-free centralized wireless network \cite{jin2015resource}. In \cite{kontovasilis1997effective}, a semi-Markovian server model was developed in a loss-less environment, based on which the effective bandwidth was proved to satisfy that the spectral radius of an appropriate nonnegative matrix is equal to unity. This result was extended to a simplified homogeneous WiFi network with error-free channels in \cite{28}, where a semi-Markovian model, expanding the classical Markov chain of WiFi \cite{840210}, generated the aforementioned nonnegative matrix (as specified in \cite{kontovasilis1997effective}). Unfortunately, the semi-Markovian model is unable to capture lossy wireless channels which require distinctive definitions of states and model structures. Moreover, the extension of the model of \cite{840210} to heterogeneous LAA networks is non-trivial.

In this paper, the effective capacity is derived to capture the distinctive properties in the distributed and heterogeneous network environments in the unlicensed band. The properties include transmission collisions, exponential backoff, and lossy wireless channels. To the best of our knowledge, the effective capacity has never been studied heretofore.

\section{System Model}

Consider $N$ LAA-BSs, and $M$ WiFi nodes, all operating in an unlicensed frequency band with a bandwidth of $B$ (in Hertz). All nodes are randomly placed, since both the LAA-BSs and WiFi access points (WiFi-APs) can be deployed in an uncoordinated, ad-hoc manner. We assume that there is no hidden node problem. Nevertheless, this assumption can be lifted, as will be discussed in Section IV.

As specified in 3GPP TR 36.889 \cite{3}, two different channel access schemes are considered for LAA, i.e., LBT with random back-off in a fixed CW, and LBT with random back-off in an exponentially increasing CW, as illustrated in Fig. 1(a) and 1(b), respectively. These two methods are referred to  as FCW (Fixed CW) and VCW (Variable CW), respectively.

In the case of FCW, a LAA-BS senses the unlicensed band for a predefined period, termed ``channel clear assess (CCA)'', whenever it has packets to transmit. If the channel is free over CCA, the LAA-BS sets an integer backoff timer randomly and uniformly within a fixed CW $[0,W_L)$, where $W_L$ is the initial CW size. The backoff timer counts down one per timeslot. It freezes if the channel is busy, and does not resume until the channel is sensed free for CCA again. Once the backoff timer turns to zero, a (re)transmission is triggered.
If the (re)transmission is collided (i.e., no acknowledgment (ACK) is returned), the LAA-BS resets the backoff timer within $[0,W_L)$ to retransmit the packet.

In the case of VCW, exponential backoff is adopted on top of FCW, where the CW doubles, each time the (re)transmission of the LAA-BS is collided (i.e., no ACK is returned). $K_L$ is the maximum number of retransmissions per packet, after which the CW is reset to $W_L$. In this sense, VCW is analogous to the distributed coordination function (DCF) of WiFi~\cite{840210}. However, the LAA-BS resets the CW to $W_L$, after a collision-free (re)transmission (i.e., neither ACK or non-ACK is returned), as opposed to the DCF. This is due to the fact that the LAA-BS can exploit OFDMA to multiplex signals for multiple LAA-UEs. It is possible that only some of the LAA-UEs succeed and return ACKs after a collision-free (re)transmission~\cite{3}. Resetting the CW can prevent the CW from continuously enlarging and staying large.

\begin{figure}[!t]
\centering
\includegraphics[width=3.5in]{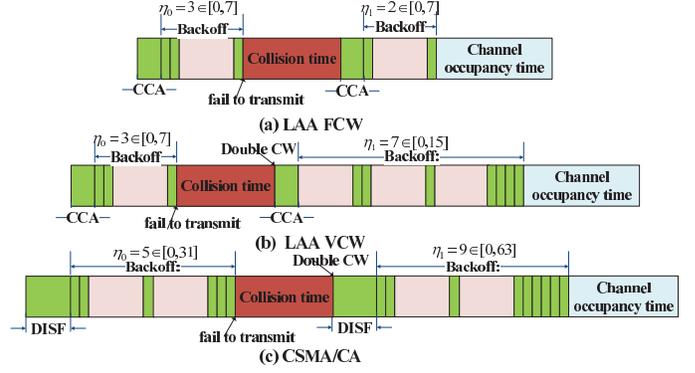}
\caption{Medium Access Mechanism for Unlicensed Spectrum: a) LAA FCW b) LAA VCW c) WiFi CSMA/CA, where the initial CW of LAA is 8, the initial CW of WiFi is 32, and ${\eta _j}$ is the number of backoff timeslots in response to the $j$-th collided (re)transmission, $j=0,\cdots, K_L-1$ (or $K_W-1$).}
\label{fig.1}
\end{figure}

For the WiFi nodes, we consider the DCF which involves both CSMA/CA and binary exponential backoff \cite{22}, as illustrated in Fig. 1(c). The initial CW size of WiFi transmissions is denoted by $W_0$, and the maximum number of retransmissions per packet is $K_W$. The period, for which a WiFi node keeps sensing the free channel before setting a backoff timer, is named ``distributed inter-frame space (DIFS)". We assume that WiFi has no QoS requirements, and all the WiFi nodes are bi-directional.

Consider the downlink. The overall transmit power of a LAA-BS is $P_{\rm tot}$. The power is allocated to $K$ LAA-UEs associated with the LAA-BS. The transmit power to LAA-UE $k$ is $P_k$. For illustration convenience, we consider that each LAA-UE $k$, $k=1,\cdots,K$, is evenly allocated a subband with the bandwidth of $\frac{B}{K}$. As a result, the instantaneous transmit rate $R_k$ of LAA-UE $k$ is given by
\begin{equation}\label{eq11}
{R_k} = \frac{B }{K}{\log _2}(1 + \frac{G_kP_k}{\sigma^2} ),
\end{equation}
where $G_k$ is the channel gain of LAA-UE $k$, and $\sigma^2$ is the noise power. It is noteworthy that the effective capacity is on a user basis, and can be quantified given the QoS requirement of a user, and the transmit power and bandwidth allocated to the user. The bandwidths can be unequal among users.

QoS provisioning is crucial to LAA systems which are integral part of the 5G networks. LAA-UEs can have different QoS requirements, such as end-to-end delay comprised of the queueing delay and transmission delay. A set of FIFO queues are used to buffer data traffic destined for different LAA-UEs, one queue per user. This is reasonable in the presence of QoS, since the backlogs of the FIFO queues indicate the queuing delays of the users. Considering the distributed network environment of LAA. The QoS can be characterized statistically by employing the QoS exponent ${\theta _k}$, $k=1,\cdots,K$, as given by \cite{20,Chang1994}
\begin{equation}\label{1a1}
\theta_k= - \mathop {\lim }\limits_{{Q_k^{\rm th}} \to \infty } \frac{{\log (\Pr \{ Q_k(\infty ) > {Q_k^{\rm th}}\} )}}{{{Q_k^{\rm th}}}},
\end{equation}
where $Q_k(t)$ is the length of the FIFO queue at the corresponding LAA-BS to buffer the downlink traffic for LAA-UE $k$ at time $t$, $Q_k^{\rm th}$ is the threshold of the queue length specified for the traffic, and $\Pr \{ Q_k(\infty ) > {Q_k^{\rm th}}\}$ is the buffer-overflow probability. In this sense, $\theta_k$ provides the exponential decaying rate of the probability that the threshold is exceeded.

The effective capacity of LAA-UE $k$, denoted by $C_k(\theta_k)$, specifies the maximum, consistent, steady-state arrival rate at the input of the FIFO queue, as given by \cite{20,21}
\begin{equation}\label{effective capacity}
C_k(\theta_k ) = - \mathop {\lim }\limits_{t \to \infty } \frac{1}{{\theta_k t}}\log (\mathbb{E}\big\{{e^{ - \theta_k S_k(t)}}\big\} ).
\end{equation}
where $S_k(t)$ is the number of bits successfully delivered to LAA-UE $k$ during $(0,t]$, and $\mathbb{E}\{\cdot\}$ denotes expectation.

Note that our proposed model is unrestricted to a particular link direction, and can be readily applied the uplink. This is because traffic flows of the uplink
and downlink are separately handled and processed in the physical and MAC layers (although
the contents of the flows might be relevant in the application layer). The QoS requirements are decoupled between the uplink and downlink, and so would be the analysis based on our model.

\section{Analysis of The Effective Capacity of LAA}
In this section, we analyze the effective capacity of LAA-UE $k$, given the QoS exponent $\theta_k$, instantaneous transmit rate $R_k$ and the number of LAA and WiFi devices $N$ and $M$. First, we put forward a new theorem to characterize the effective capacity, as follows.

\begin{thm}\label{Theorem 1}
The transmission collisions, random backoffs, and lossy wireless channels of LAA can be precisely characterized by a four-state semi-Markovian model. Given $\theta_k>0$, $k=1,\cdots,K$, the effective capacity of LAA-UE $k$, $C_k$, satisfies
\begin{equation}\label{eq29}
\begin{aligned}
(1 - p_L^{{K_L}})&{{\hat t}_1}({e^{{\theta _k}{C_k}}})\Big[{e^{( - {R_k}{\theta _k} + {\theta _k}{C_k}){T_{\rm f}}}}(1 - {\varepsilon _k})\\
& + {e^{{\theta _k}{C_k}{T_{\rm f}}}}{\varepsilon _k}\Big] + p_L^{{K_L}}{{\hat t}_2}({e^{{\theta _k}{C_k}}}) = 1,
\end{aligned}
\end{equation}
where $\hat{t}_{1}(\cdot)$ and $\hat{t}_{2}(\cdot)$ are the probability generation functions (PGFs) of the total durations of backoffs for a delivered packet and those for a dropped packet, respectively; $p_L$ is the collision probability of the corresponding LAA-BS; $T_{\rm f}$ is the duration of a collision-free (re)transmission of the LAA-BS; $\varepsilon_k$ is the packet error rate (PER) of collision-free (re)transmissions of the LAA-BS to LAA-UE $k$.
\end{thm}
\begin{proof}
Any LAA-BS, such as the LAA-BS associated with LAA-UE $k$, experiences four possible states, including $(1)$ the collision-free successful (re)transmission of a packet; $(2)$ the collision-free yet unsuccessful (re)transmission of a packet (resulting from the lossy wireless channel of LAA-UE $k$); $(3)$ the backoffs and collided (re)transmissions of such a packet until its collision-free (re)transmission; and $(4)$ the backoffs and collided (re)transmissions of a packet that exhausts all retransmissions with collisions.

The LAA-BS transits between the four states. The LAA-BS can tell the first state from the second state based on the ACK/NACK. If an ACK is returned from LAA-UE $k$, the (re)transmission is collision-free and successful; if a NACK is returned from LAA-UE $k$, the (re)transmission is collision-free yet unsuccessful. All the backoffs and collided (re)transmissions of the packet prior to the collision-free (re)transmission belong to the third state. In the case that all the (re)transmissions of a packet are collided, the LAA-BS can become aware of this since no ACK/NACK is returned. Such classification of states can fully capture the behaviors of the LAA-BS, as well as the impact of the distributed wireless environment on the LAA-BS.

A new four-state semi-Markovian model can be developed to precisely characterize the behavior of the LAA-BS associated with LAA-UE $k$ in response to the distributed wireless environment. As shown in Fig. \ref{fig.2}, the ON state corresponds to the successful transmission of a packet to LAA-UE $k$. The OFF$_1$ state corresponds to the collision-free yet unsuccessful transmission of a packet, due to the lossy wireless channel. The OFF$_2$ state corresponds to the backoffs and collided (re)transmissions of a packet before its collision-free (re)transmission. The OFF$_3$ state corresponds to the backoffs and (re)transmissions of a packet that exhausts all (re)transmissions with collisions and hence drops.

The transition probabilities between the four states are also given in Fig. \ref{fig.2}. Here, the transition probability from the ON, OFF$_1$, or OFF$_3$ state to the OFF$_2$ state is $(1 - p_L^{{K_L}})$ as is the probability that the next packet of the LAA-BS gets transmitted collision-free. $p_L^{K_L}$ is the packet drop probability after $K_L$ collided retransmissions. The transition probability from the ON, OFF$_1$, or OFF$_3$ state to the OFF$_3$ state is $p_L^{{K_L}}$, as is the probability that the next packet of the LAA-BS exhausts all $K_L$ (re)transmissions with collisions. The OFF$_2$ state can transit to the ON and OFF$_1$ states at the probabilities of $(1-{\varepsilon _k})$ and ${\varepsilon _k}$, respectively, depending on the channel condition between the LAA-BS and LAA-UE $k$.

The transition probability matrix of the four-state semi-Markovian model is given by
\begin{equation}
{\bf{P}} = \left[ {\begin{array}{*{20}{c}}
0&0&\varepsilon_k &{1 - \varepsilon_k }\\
{1 - p_L^{{K_L}}}&{p_L^{{K_L}}}&0&0\\
{1 - p_L^{{K_L}}}&{p_L^{{K_L}}}&0&0\\
{1 - p_L^{{K_L}}}&{p_L^{{K_L}}}&0&0
\end{array}} \right],
\end{equation}
where the rows (and columns) are structured as such that from top to bottom (and from left to right) are the OFF$_2$, OFF$_3$, OFF$_1$ and ON states. This is to facilitate evaluating the non-negative irreducibility of a subsequent matrix, as will be noted later.

Both the durations of the ON and OFF$_1$ states are $T_{\rm f}$, the transmission duration of a packet. The durations of the OFF$_2$ and OFF$_3$ states are assumed as ${t_1}$ and ${t_2}$, respectively. Each of them consists of backoffs and (re)transmissions of the current packet until the collision-free (re)transmission of the packet. In this sense, $t_1$ and $t_2$ consist of three types of timeslots: idle timeslots, the timeslots where there is a collision-free (re)transmission from either a WiFi node or another LAA-BS, and the timeslots where there is a collision  between the (re)transmissions of WiFi nodes and LAA-BSs. Note that the collided (re)transmissions of the designated LAA-BS are part of $t_1$ and $t_2$ while the collision-free (re)transmissions of the designated LAA-BS are not. $t_1$ and $t_2$ are random, depending on the number of (re)transmissions and the randomly selected backoff timer per (re)transmission.

The moment generating functions (MGFs) of $t_1$ and $t_2$ are ${M_1}(t) = \hat t_1({e^t})$ and ${M_2}(t) = \hat t_2({e^t})$, respectively, exploiting the property of PGF. The MGFs of an unsuccessful and successful collision-free transmissions are ${M_{on}}(t) =M_{OFF_1}= {e^{t{T_{\rm f}}}}$.

\begin{figure}[!t]
\centering
\includegraphics[width=3.5in]{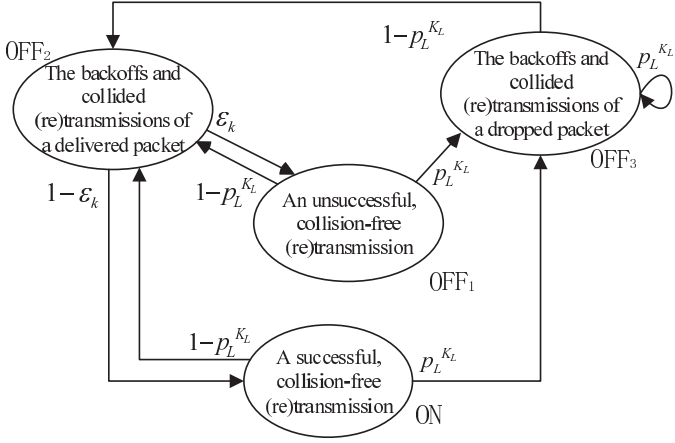}
\caption{The proposed four-state semi-Markovian model for characterizing the transmission process of a LAA-BS.}
\label{fig.2}
\end{figure}

With reference to \cite{kontovasilis1997effective,28}, we define two auxiliary variables, namely, $c$ and $u$, and construct a diagonal matrix $\mathbf{\Gamma} (c,u)$. The diagonal elements of $\mathbf{\Gamma}(c,u)$ are the MGFs of the four-state semi-Markovian model, as given by
\begin{equation}\label{eq30}
\begin{array}{l}
{\bf{\Gamma }}(c,u)=\\
 \left[ {\begin{array}{*{20}{c}}
{{M_1}( - u)}&0&0&0\\
0&{{M_2}( - u)}&0&0\\
0&0&{{M_{on}}( - u)}&0\\
0&0&0&{{M_{on}}({R_k}c - u)}
\end{array}} \right]
\end{array}.
\end{equation}
For each permissible pair of $c$ and $u$, we can write ${\bf{H}}(c,u) = {\mathbf{P}}\mathbf{\Gamma} (c,u)$\footnote{Actually this can be also written as $\bf{H}(c,u) = \mathbf{\Gamma} (c,u){\mathbf{P}}$, because the eigenvalue of both are the same.}, as given in (\ref{201611151}).
\newcounter{mytempeqncnt}
\begin{figure*}
\normalsize
\setcounter{mytempeqncnt}{\value{equation}}
\setcounter{equation}{6}
\begin{equation} \label{201611151}
{\bf{H}}(c,u) = {\mathbf{P}}\mathbf{\Gamma} (c,u)=
\left[ {\begin{array}{*{20}{c}}
0&0&{\varepsilon_k {M_{on}}( - u)}&{(1 - \varepsilon_k ){M_{on}}({R_k}c - u)}\\
{(1 - p_L^{{K_L}}){M_1}( - u)}&{p_L^{{K_L}}{M_2}( - u)}&0&0\\
{(1 - p_L^{{K_L}}){M_1}( - u)}&{p_L^{{K_L}}{M_2}( - u)}&0&0\\
{(1 - p_L^{{K_L}}){M_1}( - u)}&{p_L^{{K_L}}{M_2}( - u)}&0&0
\end{array}} \right].
\end{equation}
\setcounter{equation}{\value{mytempeqncnt}}
\setcounter{equation}{7}
\begin{equation} \label{square matrix1}
\begin{aligned}
\big |{\bf{H}}( - {\theta _k},&- {\theta _k}{C_k}) - \phi ( - {\theta _k}, - {\theta _k}{C_k}){\bf{I}}\big |\\
&= \left| {\begin{array}{*{20}{c}}
{ - \phi ( - \theta_k , - \theta_k {C_k}) }&0&{\varepsilon_k {M_{on}}( - u)}&{(1 - \varepsilon_k ){M_{on}}({R_k}c - u)}\\
{(1 - p_L^{{K_L}}){M_1}( - u)}&{p_L^{{K_L}}{M_2}( - u) - \phi ( - \theta_k , - \theta_k {C_k}) }&0&0\\
{(1 - p_L^{{K_L}}){M_1}( - u)}&{p_L^{{K_L}}{M_2}( - u)}&{ - \phi ( - \theta_k , - \theta_k {C_k}) }&0\\
{(1 - p_L^{{K_L}}){M_1}( - u)}&{p_L^{{K_L}}{M_2}( - u)}&0&{ - \phi ( - \theta_k , - \theta_k {C_k}) }
\end{array}} \right|\\
& = \phi {( - {\theta _k}, - {\theta _k}{C_k})^4} - \phi {( - {\theta _s}, - {\theta _k}{C_k})^2}{M_1}({\theta _k}{C_k})(1 - p_L^{{K_L}})\left( {{e^{( - {\theta _k}{R_k}{\rm{ + }}{\theta _k}{C_k}){{\rm{T}}_{\rm f}}}}(1 - {\varepsilon _k}) + {e^{({\theta _k}{C_k}){{\rm{T}}_{\rm f}}}}{\varepsilon _k}} \right)\\
&~~~- \phi {( - {\theta _k}, - {\theta _k}{C_k})^3}p_L^{{K_L}}{M_2}({\theta _k}{C_k})= 0.
\end{aligned}
\end{equation}
\setcounter{equation}{\value{mytempeqncnt}}

\hrulefill %
\vspace*{4pt}
\end{figure*}

\setcounter{equation}{8}

Note that $\mathbf{H}(c,u)$ is non-negative irreducible, as it cannot be rearranged to an upper-triangular matrix, e.g., by using the Gaussian-Newton method. As a result, the spectral radius of $\mathbf{H}(c,u)$, denoted $\phi (c,u) = \rho \big(\mathbf H(c,u)\big)$ is a simple eigenvalue of $\mathbf{H}(c,u)$, where  $\rho(\cdot)$  denotes spectral radius.

By [29, Theorem 3.1], given $c \leq 0$, there exists a unique $u^*(c)$ such that $ \phi(c,u^*(c))=1$ and $ \mathop {\lim }\limits_{t \to \infty } \frac{1}{t}\log ({\mathbb{E}}\{ {e^{ {c}{S_k}(t)}}\} )=u^*(c)$.
By [29, Theorem 3.2], the effective capacity $C_k = \frac{u(c)}{c}$ when $\phi(c,u(c))=1$ and $\theta_k = -c$.
As a result, the effective capacity $C_k$ can be evaluated by solving $\phi ( - \theta_k , - \theta_k {C_k}) = 1$ for $\theta_k  > 0$ \cite{28}. Since $\phi ( - \theta_k , - \theta_k {C_k})$ is an eigenvalue of $\mathbf{H}( - \theta_k , - \theta_k {C_s})$ , we can have (\ref{square matrix1}), where $\mathbf{I}$ is the identity matrix, and $|\cdot|$ stands for determinant.

Finally, we substitute $\lambda=\phi(-\theta_k,-\theta_kC_k)=1$ into (\ref{square matrix1}), and obtain (\ref{eq29}). This concludes the proof.
\end{proof}

As an input to the proposed semi-Markovian model, $p_L$ can be calculated in prior in the absence of hidden nodes, as described in Appendix A. The hidden node problem would only affect the value of the input. It would not affect our model and the way that the model analyzes the effective capacity of LAA-BSs, given $p_L$. In the presence of hidden nodes, $p_L$ can be calculated by extending existing studies on WiFi hidden node~\cite{Ekici2008,Jang2012}, as briefly discussed in Appendix A.

The PGFs, $\hat t_1(z)$ and $\hat t_2(z)$, can be derived, as required in Theorem 1. By the law of total expectation \cite{mendenhall2012introduction}, $\hat t_1(z)$ can be given by
\begin{equation}\label{eq12}
{\hat t_1}(z) = {\mathbb{E}_i}\big \{\mathbb{E}\{{z^{t_1^{(i)}}}\big |i \leq {K_L}-1\}\big\},
\end{equation}
where $i$ is the number of collisions as per a packet; $\mathbb{E}_i\{\cdot\}$ takes the expectation over $i\leq K_L$; $t_1^{(i)}$ accounts for the $i$ collided (re)transmission processes and can be written as
\begin{equation}\label{gy1122}
t_1^{(i)} = i{T_{\rm c}} + \sum\limits_{d = 1}^{{I_i}} {{\tau_d}}, 0 \le i \le {K_L} - 1,
\end{equation}
where $T_{\rm c}$ is the duration of a collided (re)transmission of the LAA-BS, $I_i$ is the total number of timeslots that the LAA-BS has backed off in response to the $i$ collisions, and $\tau_d$ is the duration of the $d$-th timeslot since the successful (re)transmission of the last packet, $d=1,\cdots,I_i$.

 We can write $I_i$ as
\begin{equation}\label{eq13}
{I_i}{=}\sum\limits_{j = 0}^i {{\eta _j}},   0 \le i \le {K_L} - 1,
\end{equation}
where ${\eta _j}$ is the number of timeslots in response to the $j$-th collided (re)transmission of the LAA-BS. By definition, the PGF of ${I_i}$ is ${\hat I_i}(z) = \prod\limits_{j = 0}^i {{{\hat \eta }_j}(z)}$. \footnote{In the case of FCW, ${\eta _j}$ is uniformly distributed within $[0,{W_L} - 1]$. ${\hat \eta _j}(z)$ is given by
${\hat \eta _j}(z) = \frac{1}{{{W_L}}}\frac{{1 - {z^{{W_L}}}}}{{1 - z}}.$
In the case of VCW, ${\eta _j}$ is uniformly distributed within $[0,{2^j}W_{L } - 1]$. ${\hat \eta _j}(z)$ is ${\hat \eta _j}(z) = \frac{1}{{{2^j}{W_{L }}}}\frac{{1 - {z^{{2^j}{W_{L}}}}}}{{1 - z}}.$}

Consider that $\tau_d$ is independent and identically distributed, and $I_i$ is independent, discrete random variable taking non-negative integer values. We suppress the subscript ``$_d$''. Using the law of total expectation \cite{mendenhall2012introduction}, the PGF of $\sum\limits_{d = 1}^{{I_i}}\tau_d$ can be given by
\begin{equation}\label{eq133}
\begin{aligned}
\mathbb{E}_{I_i}\big\{{z^{\sum\limits_{d = 1}^{{I_i}} {{\tau_d}} }}\big\}&= \mathbb{E}\big\{\mathbb{E}\big\{{z^{\sum\limits_{d = 1}^{{I_i}} {{\tau_d}} }}|{I_i}\big\}\big\} = \mathbb{E}\big\{{(\hat \tau(z))^{{I_i}}}\big\}\\
& = {{\hat I}_i}(\hat \tau(z)) = \prod\limits_{j = 0}^i {{{\hat \eta }_j}(\hat \tau(z))},
\end{aligned}
\end{equation}
where $\hat \tau(z)$ is the PGF of $\tau$, as given in Appendix B.

Substitute (\ref{gy1122}), (\ref{eq13}) and (\ref{eq133}) into (\ref{eq12}),
\begin{equation}\label{gy14}
\begin{aligned}
{{\hat t}_1}(z) &= \frac{{\sum\limits_{i = 0}^{{K_L} - 1} \bigg[{(1 - {p_L})p_L^i \mathbb{E}({z^{t_1^{(i)}}})\bigg]} }}{{1 - p_L^{{K_L}}}}\\
&= \frac{{\sum\limits_{i = 0}^{{K_L} - 1} {\bigg[(1 - {p_L})p_L^i{z^{i{T_{\rm c}}}}\prod\limits_{j = 0}^i {{{\hat \eta }_j}\Big(\hat \tau(z)\Big)}\bigg]} }}{{1 - p_L^{{K_L}}}},
\end{aligned}
\end{equation}
where $p_L$ can be calculated in Appendix A.

When the retransmission attempt $K_L$ is reached, the packet will be dropped. $t_2$ can be written as
\begin{equation}
t_2 = K_L{T_{\rm c}} + \sum\limits_{j = 1}^{{I_{K_L-1}}} {{\tau_j}}.
\end{equation}
As a result, $\hat t_2(z)$ can be given by
\begin{equation}\label{gy15}
{{\hat t}_2}(z) = \mathbb{E}\big\{{z^{t_2}}|i = {K_L}\big\}= {z^{{K_L}{T_{\rm c}}}}\prod\limits_{j = 0}^{{K_L-1}} {{{\hat \eta }_j}(\hat \tau(z))}.
\end{equation}

\textit{Remark:} Theorem 1 provides an accurate analysis for the effective capacity of LAA through the new precise four-state semi-Markovian model. The calculations of the PGFs $\hat{t}_1(z)$ and $\hat{t}_2(z)$, tailored for the theorem, also play an important role in the design and optimization of LAA, as will be shown in Section V-A.  However, (\ref{eq29}) is an implicit function of both $C_k$ and $\theta_k$, which is hardly conducive to revealing the intrinsic connections between $C_k$ and $\theta_k$. To uncover the connections and shed insights, we proceed to develop a more abstract model. Validated by Corollary 1, the abstract model is equivalent to the four-state Markovian model and provides accurate analysis. More importantly, the abstract model provides a key step to reveal the concavity of the effective capacity, which allows the effective capacity to be maximized through structured optimization.

\begin{cor}\label{cor 1}
Given $\theta_k$, the effective capacity of a LAA-BS can be also evaluated by exploiting a two-state semi-Markovian model without compromising modelling accuracy.
\end{cor}
\begin{proof}
Fig. \ref{fig.3} shows the two-state semi-Markovian model, where the ON state corresponds to a successful (re)transmission of the LAA-BS associated with LAA-UE $k$, and the OFF state indicates the intervals between any two consecutive successful (re)transmissions. The OFF state covers the OFF$_1$, OFF$_2$ and OFF$_3$ states in the four-state semi-Markov model described in the proof of Theorem 1. We proceed to show that the two-state model provides the same analysis of the effective capacity as Theorem 1.

The transition probabilities between the ON and OFF states are one in both directions of the two-state semi-Markov model. This is due to the definition of the states. The transition probability matrix is ${\bf{P}} = \left[ {\begin{array}{*{20}{c}}
{\rm{0}}&{1}\\
1&0
\end{array}} \right]$.

\begin{figure}[!t]
\centering
\includegraphics[width=2.5in]{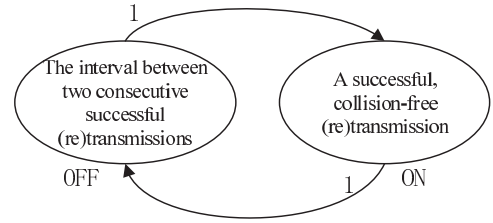}
\caption{The On-Off semi-Markovian model.}
\label{fig.3}
\end{figure}

Let $t_3$ denote the duration of the OFF state. The MGF of $t_3$ is ${M_{3}}(t) = \hat t_3({e^t})$. Also define two auxiliary variables $c$ and $u$, as done in the proof of Theorem 1, and construct $\mathbf{\Gamma} (c,u)$:
\begin{equation}\label{eq30}
\begin{aligned}
\mathbf{\Gamma} (c,u) &= \left[ {\begin{array}{*{20}{c}}
{{M_{3}}(- u)}&0\\
0&{{M_{on}}(R_k  c - u)}
\end{array}} \right]\\
&= \left[ {\begin{array}{*{20}{c}}
{\hat t^3_{off}({e^{ - u}})}&0\\
0&{{e^{(R_k c - u){T_{sl}}}}}
\end{array}} \right].
\end{aligned}
\end{equation}

For each permissible pair of $c$ and $u$, we can write \cite{28}
\begin{equation}
{\bf{H}}(c,u) = {\bf{P}}{\bf{\Gamma }}(c,u)= \left[ {\begin{array}{*{20}{c}}
{\rm{0}}&{{e^{({R_k}c - u){T_{\rm f}}}}}\\
{{{\hat t}_3}({e^{ - u}})}&0
\end{array}} \right].
\end{equation}
The effective capacity of LAA-UE $k$ can be evaluated by solving $\phi ( - \theta_k , - \theta_k {C_k}) = 1$ for $\theta_k > 0$ \cite{28}, where $\phi(c,u)$ is the eigenvalue of $\mathbf{H}(c,u)$.

As an eigenvalue of $\mathbf{H}( - \theta_k , - \theta_k {C_k})$, $\phi(-\theta_k,-\theta_kC_k)$ satisfies
\begin{equation}\label{square matrix}
\begin{aligned}
\big|{\bf{H}}&( - {\theta _k},- {\theta _k}{C_k}) -  \phi ( - \theta_k , - \theta_k {C_k}) {\bf{I}} \big|\\
&= \left[ {\begin{array}{*{20}{c}}
{ -  \phi ( - \theta_k , - \theta_k {C_k}) }&{{e^{({R_k}c - u){T_{\rm f}}}}}\\
{{{\hat t}_3}({e^{ - u}})}&{ -  \phi ( - \theta_k , - \theta_k {C_k}) }
\end{array}} \right]\\
&= { \phi ( - \theta_k , - \theta_k {C_k}) ^2} - {e^{ - {\theta _k}({R_k} - {C_k}){T_{\rm f}}}}{{\hat t}_3}({e^{{\theta _k}{C_k}}})= 0.
\end{aligned}
\end{equation}

Note that $t_3$ can consist of multiple backoffs and (re)transmissions of a number of dropped packets undergoing $K_L$ collided (re)transmissions and a number of collision-free but unsuccessful packets incurring poor channel conditions. Let $G_1$ denote the number of collision-free but unsuccessful packets, and $G_{2,a}$ denote the number of dropped packets between the $(a-1)$-th and $a$-th collision-free (re)transmissions. $a=1,\cdots,G_1$. The $0$-th is the former of the two consecutive successful (re)transmissions spanning $t_3$.

As a result, $t_3$ can be given by
\begin{equation*}
{t_3} = t_{1,G_1+1} + \sum\limits_{b=0}^{{G_{2,G_1+1}}} {{t_{2,b,G_1+1}}} {\rm{ + }}\sum\limits_{a=1}^{{G_1}} {(t_{1,a} + {T_{sl}} + \sum\limits_{b=0}^{{G_{2,a}}} {{t_{2,b,a}}} )},
\end{equation*}
where $t_{1,a}$, $a=1,\cdots,G_1$, is the duration of all the backoffs and collided (re)transmissions for the $a$-th collision-free yet unsuccessful packet; $t_{1,G_1+1}$ is the duration of all the backoffs and collided (re)transmissions for the latter of the two consecutive successful (re)transmissions; $t_{2,b,a}$, $a=1,\cdots,G_1$, $b=0,\cdots, G_{2,a}$, is the duration of all $K_L$ backoffs and collided (re)transmissions for the $b$-th collided packet between the $(a-1)$-th and $a$-th collision-free (re)transmissions.

Also note that $t_{1,a}$, $a=1,\cdots, G_1$, are independent and identically distributed, and all yield the PGF $\hat{t}_1(z)$. Likewise, $t_{2,b,a}$, $a=1,\cdots, G_1$ and $b=0,\cdots, G_{2,a}$, all yield the PGF $\hat{t}_2(z)$. From (\ref{gy14}) and (\ref{gy15}), we have
\begin{equation} \label{gyb}
\begin{aligned}
{{\hat t}_3}&(z)= \mathbb{E}\{{z^{{t_3}}}\}\\
=&(1 - \varepsilon_k )(1 - p_L^{{K_L}}){{\hat t}_1}(z)\sum\limits_{{G_{2,G_1+1}} = 0}^\infty  {{{\left( {p_L^{{K_L}}{{\hat t}_2}(z)} \right)}^{{G_{2,G_1+1}}}}} \times \\
&\sum\limits_{{G_1} = 0}^\infty  {{{\left( {\varepsilon_k (1 - p_L^{{K_L}}){{\hat t}_1}(z){z^{{T_{\rm f}}}}\sum\limits_{{G_{2,G_1}} = 0}^\infty  {{{\left( {p_l^{{K_L}}{{\hat t}_2}(z)} \right)}^{{G_{2,G_1}}}}} } \right)}^{{G_1}}}} \\
=& \frac{{(1 - \varepsilon_k )(1 - p_L^{{K_L}}){{\hat t}_1}(z)}}{{1 - p_L^{{K_L}}{{\hat t}_2}(z) - \varepsilon_k (1 - p_L^{{K_L}}){{\hat t}_1}(z){z^{{T_{\rm f}}}}}},
\end{aligned}
\end{equation}
where the last equality is obtained by using the sum formula for geometric progressions.

Substituting $\phi(-\theta_k,-\theta_kC_k)=1$ into (\ref{square matrix}), we obtain
\begin{equation} \label{gya}
{e^{( - {R_k}{\theta _k} + {\theta _k}{C_k}){T_{\rm f}}}}{{\hat t}_3}({e^{{\theta _k}{C_k}}}) = 1.
\end{equation}
Substituting (\ref{gyb}) into (\ref{gya}), and then rearranging (\ref{gya}), we can finally obtain (6). Corollary 1 is proven.

\end{proof}

Corollary 1 dictates that the two-state semi-Markovian model can accurately capture the effective capacity of LAA. A key step of the two-state semi-Markovian model, i.e., (\ref{gya}), sheds important insights on the design and optimization of LAA. It can be used to reveal the strict concavity of the effective capacity, as stated in the following theorem.
\begin{thm}\label{Theorem 2}
Given $\theta_k>0$, the effective capacity of LAA-UE $k$, $C_k$, is concave in the transmit rate $R_k$ and can be given by
\begin{equation} \label{gyy1}
{C_k} = \frac{{{F^{ - 1}}({R_k}{\theta _k}{T_{\rm f}})}}{{{\theta _k}}},
\end{equation}
where $F(x) = \log (\hat t_3({e^x})) + x{T_{\rm f}}$, and ${F^{ - 1}}(\cdot)$ is the inverse function of $F(\cdot)$.

\end{thm}
\begin{proof}
Taking the logarithm at the both sides of (\ref{gya}), we can have
\begin{equation}
\log \Big({{\hat t}_3}({e^{{\theta _k}{C_k}}})\Big) + {\theta _k}{C_k}{T_{\rm f}} = {R_k}{\theta _k}{T_{\rm f}},
\end{equation}
where the left-hand side (LHS) is $F(C_k\theta_k)$. This confirms (\ref{gyy1}).

To prove the concavity of $C_k$ in $R_k$, we first prove that $F(\cdot)$ is convex. This is because
\begin{equation}\label{eq41}
\begin{aligned}
&\alpha F({x_1}) + (1 - \alpha )F({x_2})= \alpha \log ({{\hat t}_3}({e^{{x_1}}})) + \alpha {x_1}{T_{\rm f}} \\
 &~~~+ (1 - \alpha )\log ({{\hat t}_3}({e^{{x_2}}})) + (1 - \alpha ){x_2}{T_{\rm f}}\\
 &= \log (\mathbb{E}{\left( {{e^{{x_1}{t_3}}}} \right)^\alpha }\mathbb{E}{\left( {{e^{{x_2}{t_3}}}} \right)^{1 - \alpha }}) + (\alpha {x_1} + (1 - \alpha ){x_2}){T_{\rm f}}.
\end{aligned}
\end{equation}
Applying Lyapunov inequality \cite{32},
\begin{equation}\label{eq42}
\left(\mathbb{E}\Big\{ {{{({e^{{t_3}}})}^{{x_1}}}} \Big\}\right)^{\alpha }\left(\mathbb{E}\Big\{ {{{({e^{{t_3}}})}^{{x_2}}}} \Big\}\right)^{1 - \alpha } \ge \mathbb{E}\Big\{ {{({e^{{t_3}}})}^{\alpha {x_1} + (1 - \alpha ){x_2}}}\Big\}.
\end{equation}
Substituting (\ref{eq42}) in (\ref{eq41}), we can have
\begin{equation}\label{eq43}
\begin{array}{l}
\alpha F({x_1}) + (1 - \alpha )F({x_2})\\
 \ge \log \left(\mathbb{E}\Big\{{e^{(\alpha {x_1} + (1 - \alpha ){x_2}){t_{3}}}}\Big\}\right) + (\alpha {x_1} + (1 - \alpha ){x_2}){T_{\rm f}}\\
 = F(\alpha {x_1}{\rm{ + }}(1 - \alpha ){x_2}).
\end{array}
\end{equation}
As a result, $F(x)$ is convex.

By the definition of PGF, $\hat t_3(z)$ is strictly monotonically increasing with $z\geq 1$. $e^x$ is also monotonically increasing. Therefore, $\hat t_3({e^x})$ is strictly and monotonically increasing. In turn, $F(x)$ is strictly and monotonically increasing. Given $\theta_k$, ${F^{ - 1}}({R_k}{\theta _k}{T_{\rm f}})$ is concave in $R_k$, and so is ${C_k}$.
\end{proof}

\section{Applications of the effective capacity of LAA}
Our proposed theorems and corollary can have important applications in the design and control of LAA systems.
\subsection{Maximization of Effective Capacity}
One of the applications is to maximize the effective capacity of LAA. Recall that a LAA-BS equally allocates the bandwidth $B$ to $K$ LAA-UEs. The LAA-BS can optimally control its transmit powers for the LAA-UEs, to maximize the total effective capacity, given $\theta_k$.

From (\ref{gyy1}), we show that $C_k$ is a function of $R_k$ and in turn, a function of $P_k$, i.e., $C_k(P_k)$. The maximization of the total effective capacity can be formulated as
\begin{subequations} \label{eq35}
\begin{equation}\label{eq35a} \textbf{{P1:}} \mathop {{\max}}\limits_{{{P_k\,\forall k}}} \sum\limits_{k=1}^K {C_k(P_k)}\end{equation}
\begin{equation}\label{eq35b} ~~~~~~~~~{s.t.} \sum_{k=1}^K {P_k}  \le {P_{\rm tot}},\end{equation}
\begin{equation}\label{eq35c}~~~~~~P_k \ge 0, k=1,\cdots,K.\end{equation}
\end{subequations}
where (\ref{eq35b}) restricts the total transmit power, and (\ref{eq35c}) constrains the transmit powers to be non-negative.

Rewriting (\ref{gyy1}) as ${R_k}{=}\frac{{F(C_k{\theta _k})}}{{{\theta _k}{T_{\rm f}}}}$ and substituting (\ref{eq11}) into it, we can write $P_k$ as a function of $C_k$, as given by
\[
P_k({C_k}) = \frac{{{\sigma ^2}}}{{{G_k}}}\left( {\exp \left( {\frac{{F({C_k}{\theta _k})K\ln 2}}{{B {\theta _k}{T_{\rm f}}}}} \right) - 1} \right),
\]
where $F(\cdot)$ can be explicitly rewritten by substituting (\ref{gyb}), then (\ref{gy14}) and (\ref{gy15}). In other words, the calculations of the PGFs $\hat{t}_1(z)$ and $\hat{t}_2(z)$, tailored for Theorem 1, are important to evaluate and optimize $C_k$. As a result, (P1) is reformulated as
\begin{subequations} \label{eq38}
\begin{equation}\textbf{{P2:}} \mathop {{\max}}\limits_{{{C_k\,\forall k}}} \sum\limits_{k=1}^K {C_k}\end{equation}
\begin{equation}\label{eq38b}~~~~~~~~~~~~~{s.t.} \sum\limits_{k=1}^K {P_k(C_k)}  \le {P_{\rm tot}},\end{equation}
\begin{equation}~~~~~~~~~~~~~~~~~~~~~C_k \ge 0 , k=1,\cdots,K.\end{equation}
\end{subequations}

Exploiting Theorem 2, we can prove that (P2) is convex and holds strong duality.
\begin{proof}
From Theorem 2, $F(x)$ is convex in $x$. By the composition rules of optimization, $P_k(C_k)$ is convex in $C_k$. Given the linear objective and the convex constraints, (P2) is convex.

Further, we can show that the point $C_k^* = 0, k=1,\cdots,K,$ belongs to the feasible region of the problem, i.e.,
\begin{equation}\label{eq48}
\sum\limits_{k=1}^K {P_k(C_k^*)}  < {P_{\rm tot}},
\end{equation}
\begin{equation}\label{eq49}
C_k^* \ge 0 ,k=1,\cdots,K.
\end{equation}
By the Slater's condition \cite{boyd2004convex}, (P2) holds strong duality. This concludes the proof.
\end{proof}

Unfortunately, (\ref{eq38}) cannot be structured to conform to a standard input of popular convex tools, such as MATLAB cvx toolbox, due to $\log\Big(\hat{t}_3(e^{C_k\theta_k})\Big)$ in (\ref{eq38b}). We propose to solve (\ref{eq38}) by taking Lagrange dual decompositions.

The dual problem of (\ref{eq38}) is given by
\begin{equation}\label{eq46}
\mathop{\max }\limits_\mu  \mathop {\min }\limits_{{{\bf{C_k}}}} \mathcal{L}(\mu ,{{C}}_k),
\end{equation}
where $\mu \geq 0$ is the dual Lagrange multiplier for (\ref{eq38b}), and $\mathcal{L}(\mu, C_k)$ is the Lagrange function, as given by
\begin{equation}\label{Lagrange function}
\mathcal{L}(\mu, C_k) = -\sum\limits_{k=1}^K {C_k}  + \mu \big(\sum\limits_{k=1}^K {P_k(C_k)} -{P_{\rm tot}}\big).
\end{equation}

Given $\mu$, $C_k$ can be optimized in parallel for every LAA-UE $k$ by taking the KKT conditions of (\ref{Lagrange function}), i.e.,
\[
\mu\frac{d P_k(C_k)}{d C_k}-1=0, k = 1,\cdots, K.
\]
Since $P_k(C_k)$ is convex, $\frac{d P_k(C_k)}{d C_k}$ is monotonic. As a result, the optimal $C_k^*$ can be efficiently solved by using bisection search \cite{24}.
Given $C_k^*$, we can update $\mu $ using the subgradient method, i.e., $\mu \leftarrow {\left[ {\mu  + {\delta}\left( {\sum\limits_{k=1}^k {P_k(C_k^*) - } {P_{\rm tot}}} \right)} \right]^ + }$, where ${\left[  \cdot  \right]^ + }$ is a projection on the positive orthant, and ${\delta} >0$ is the step size.
We can repeat the bisectional search for $C_k^*$ and the subgradient update of $\mu$ until convergence. Given the strong duality of (P2), the convergent $C_k^*$, $k=1,\cdots,K$, are the solution for (\ref{eq38}).

The above optimization of transmit powers can be readily extrapolated to joint allocation of bandwidth and power. Specifically, we can optimize the transmit powers given bandwidth allocation, as described, and then adjust the bandwidths by taking a Branch and Bound (BnB) method for discrete subchannels or a block coordinated descent (BCD) method for continuous bandwidths. These two steps repeat in an alternating manner until the effective capacity is maximized.
\subsection{Maximization of Effective Energy Efficiency}
Another application of our analysis in Section IV is to maximize the effective energy efficiency of LAA (in bits/Joule). The effective energy efficiency per LAA-BS is defined to be $\frac{{\sum_{k=1}^K {{C_k}(P_k)} }}{{\bar{P} }}$, where $\bar{P} $ is the average power consumption of the LAA-BS.
$\bar{P}$ consists of the static power such as cooling, denoted by ${P_{\rm {static}}}$, the power for sensing the band, denoted by ${P_{\rm {idle}}}$, and the transmission-dependent power depending on $P_k$ \cite{19,30}.

Employing the two-state semi-Markovian model, the average power of the LAA-BS is given by
\begin{equation}\label{eq26}
\begin{aligned}
\bar P &\mathop  = \limits^{(a)}\frac{{{\pi _1}({P_{\rm{idle}}}\bar I\bar \tau + \bar i\sum\limits_{k=1}^K \frac{1}{\xi} {P_k} {T_{\rm c}}) + {\pi _2}\sum\limits_{k=1}^K \frac{1}{\xi} {P_k} {T_{\rm f}}}}{{{\pi _1}(\bar I\bar \tau + \bar i{T_{\rm c}}) + {\pi _2}{T_{\rm f}}}} + {P_{\rm{static}}}\\
& = \frac{{{\pi _1}{P_{\rm{idle}}}\bar I\bar \tau + \sum\limits_{k=1}^K \frac{1}{\xi} {P_k} ({\pi _1}\bar i{T_{\rm c}} + {\pi _2}{T_{\rm f}})}}{{{\pi _1}(\bar I\bar \tau + \bar i{T_{\rm c}}) + {\pi _2}{T_{\rm f}}}} + {P_{\rm{static}}},\\
\end{aligned}
\end{equation}
where $\pi_1$ and $\pi_2$ are the stationary probabilities of the two-state semi-Markov model in Corollary 1; $\bar \tau$ is the average duration of a timeslot and provided in Appendix B; $\xi$ is the power amplifier efficiency. For illustration convenience, we assume the power amplifier is linear and $\xi$ is constant. The numerator in $(a)$ is the total energy consumption of the ON and OFF states in the two-state semi-Markovian model. The denominator is the corresponding duration of the states.

In (\ref{eq26}), $\bar i$ is the average number of collisions between two consecutive successful
(re)transmissions. In the case of $\varepsilon_k\rightarrow 0$, $\bar i$ can be given by
\begin{equation}\label{eq16}
\bar i = \frac{p_L}{(1-p_L)^2}.
\end{equation}
\begin{proof}
$\bar{i}$ can be written as $\bar{i}=\sum_{i=1}^{\infty}ip_L^i,$ which can be further rewritten as $\bar{i}=p_L+p_L\sum_{i=1}^{\infty}(i+1)p_L^i$ and then restructured to be $\frac{\bar{i}-p_l}{p_l}=\sum_{i=1}^{\infty}(i+1)p_l^i.$ Subtracting the first and the third equations, we have $\sum_{i=1}^{\infty}p_l^i=\frac{\bar{i}-p_L}{p_L}-\bar{i},$ where the LHS is equal to $\frac{p_L}{1-p_L}$.
As a result, $\bar{i}=\frac{p_L}{(1-p_L)^2}$.
\end{proof}

In (\ref{eq26}), $\bar{I}$ is the average number of timeslots between any two consecutive successful (re)transmissions. Take FCW for example in the case of $\varepsilon_k\rightarrow 0$, $\bar I$ is given by
\begin{equation}\label{eq17}
\bar I = \sum\limits_{i = 0}^\infty  {\frac{{{W_L} + 1}}{2}p_L^i}=\frac{W_L+1}{2(1-p_L)}.
\end{equation}

Moreover, (\ref{eq26}) can be rewritten as
\begin{equation}
\bar P = {P'_{\rm{static}}} + \frac{1}{{\xi '}}\sum\limits_{k = 1}^K {{P_k}},
\end{equation}
where ${{P'}_{\rm{static}}} = {P_{\rm{static}}} + \frac{{{\pi _1}{P_{\rm{idle}}}\bar I\bar \tau}}{{{\pi _1}(\bar I\bar X + \bar i{T_{cl}}) + {\pi _2}{T_{\rm f}}}}$, and $\xi ' = \frac{{{\pi _1}(\bar I\bar \tau  + \bar k{T_{\rm{c}}}) + {\pi _2}{T_{\rm{f}}}}}{{\xi ({\pi _1}\bar i{T_{\rm{c}}} + {\pi _2}{T_{\rm{f}}})}}$.

The maximization of the effective energy efficiency can be formulated as
\begin{subequations} \label{eq49}
\begin{equation}\textbf{{P3:}} \max_{P_k\,\forall k} \frac{{\sum\limits_{k=1}^K {{C_k}(P_k)} }}{{{{P'}_{\rm{static}}} + \frac{1}{{\xi '}} \sum\limits_{k=1}^K {P_k} }}  \end{equation}
\begin{equation}~~~~~~~s.t. \sum\limits_{k=1}^K {P_k}  \le {P_{\rm tot}},\end{equation}
\begin{equation}~~~~~~~~~~~~~~~~P_k \ge 0,k=1,\cdots,K.\end{equation}
\end{subequations}

Here, (\ref{eq49}) is a fractional program, and can be readily reformulated as a parametric convex optimization problem \cite{30}. The resultant parametric convex optimization problem takes $C_k$ as variables (as done in Section V-A). By defining a non-negative auxiliary variable $\varpi $, the parametric convex problem can be formulated as
\begin{subequations} \label{eq52}
\begin{equation} \label{eq52a}
\textbf{{P4:}}~{\emph H}(\varpi ) = \mathop {{\min}}\limits_{{{\bf{C_k}}}}  \bigg\{- \sum\limits_{k \in {\bf{K}}} {C_k}  + \varpi \Big({{P'}_{\rm{static}}} + \frac{1}{{\xi '}} \sum\limits_{k=1}^K {P_k(C_k)} \Big)\bigg\}\end{equation}
\begin{equation} \label{eq52b}
s.t. \sum\limits_{k=1}^K {P_k(C_k)}  \le {P_{\rm tot}},~~~~~~~~~~\end{equation}
\begin{equation} \label{eq52c}
C_k \ge 0 ,k=1,\cdots,K.~~\end{equation}
\end{subequations}

Solving (P3) is equivalent to determining the maximum ${\varpi ^{\rm{*}}}$ satisfying ${\emph H}(\varpi^* ) = 0$. Clearly, given $C_k$, ${\emph H}(\varpi )$ is strictly and monotonically increasing with $\varpi \geq 0$  \cite{30}. We can solve ${\emph H}(\varpi )=0$ using one-dimensional search, such as the Dinkelbach's method \cite{31}, with guaranteed convergence.

By exploiting Theorem 2, $P_k(C_k)$ is proved to be convex with strong duality; see Section V-A. As a result, (\ref{eq52a}), (\ref{eq52b}) and (\ref{eq52c}) are all convex. Given $\varpi$, (P4) is convex and the Lagrangian of (\ref{eq52}) can be given by
\begin{equation}\label{eq56}
\begin{aligned}
\mathcal{L}(\mu ,{C_k}) =  &- \sum\limits_{k = 1}^K {{C_k}} + (\mu  + \frac{\varpi}{{\xi '}})\sum\limits_{k = 1}^K {f({C_k})} \\
&- \mu {P_{{\rm{tot}}}} + \varpi {{P'}_{\rm{static}}}
\end{aligned}
\end{equation}
where $\mu $ is the Lagrangian multiplier for (\ref{eq52b}).

Given $\mu $ and $\varpi $, the optimal solution for (\ref{eq56}) can be taken by solving the KKT conditions of (\ref{eq56}), i.e.,
\begin{equation}\label{eq57}
(\mu  +  \frac{\varpi}{{\xi '}})\frac{d P_k(C_k)}{d C_k}-1 = 0,
\end{equation}
where the optimal $C_k^*$, $k=1,\cdots, K$, can be achieved by using bisection search since $\frac{d P_k(C_k)}{d C_k}$ is monotonic, as discussed in Section V-A.

Given $C_k^*$ ($k=1,\cdots, K)$, $\mu$ can be updated by using the subgradient method, as described in Section V-A, and $\omega$ can be updated by using the Dinkelbach's method. This repeats until convergence. The convergent $C_k^*$, $k=1,\cdots, K$, are the global optimal solution for (\ref{eq49}).
%
%
\section{Simulations and Numerical Results}
In this section, we validate the proposed effective capacity of LAA. We also demonstrate the optimization of the transmit power to maximize the new effective capacity and effective energy efficiency. In our simulations, the LAA-BSs and WiFi nodes are randomly and uniformly distributed within an area of $500\times500$ m$^2$. The LAA-UEs associated with a LAA-BS are randomly and uniformly distributed within an area of $\frac{500}{\sqrt{N}}\times\frac{500}{\sqrt{N}}$ m$^2$ centered at the LAA-BS. The ITU-UMI model~\cite{access2010further} is used generate the channels between the LAA-UEs and the LAA-BS. $\sigma^2=-174$ dBm/Hz at each LAA-UE. $P_{\rm tot}= 23$ dBm. We assume that the LBT energy detection threshold is sufficiently low, and there is no hidden node problem. This is consistent with the assumption in the paper Section III (Nevertheless, the assumption can be lifted without changing our model, as discussed in Section IV). We also assume that all LAA-UEs can accurately measure and feed back their channel gains to the LAA-BSs via the licensed band, and the channel gains are precisely known to the LAA-BSs. This assumption is reasonable, due to the real scalar nature of chain gains and the exponentially decreasing error of scalar quantization (as quantization bits increase).

In the simulations, the traffic model is constant traffic arrival. This is because the effective capacity, specifies the maximum, consistent, steady-state input rate without violating QoS, and is of particular importance to avoid traffic saturation at the LAA-BSs. However, the proposed model can also be employed to analyze stochastic traffic arrivals. Some recent works can transform stochastic traffic arrivals into equivalent constant arrivals by using the idea of effective bandwidth \cite{Du2010}. Our proposed model can evaluate the maximum rate of these equivalent constant traffic arrivals that can be accommodated given QoS requirements, and then convert the maximum constant rate back to the one describing the stochastic arrivals.

The 5GHz unlicensed band is considered. Both the FCW and VCW modes are taken into account. $W_L=16$, $W_0=32$, $K_L=K_W=6$, and $\sigma_{\rm idle}=10\mu$s. The duration is 1ms per (re)transmission of LAA and WiFi. The durations of CCA and DIFS are $34\mu$s and $50\mu$s, respectively. These parameters are consistent with 3GPP Release-13~\cite{3}. Keep in mind that QoS has yet to be specified in the standard. Our simulations are reflective of actual LAA specified in 3GPP Release-13, but in the presence of QoS. Our simulation runs for 100 seconds at a sampling interval of 2 milliseconds, to achieve the convergent result of the effective capacity \cite{20}. Without loss of generality, we assume all the LAA-UEs have the same QoS exponent $\theta$.

\begin{figure}[!t]
\centering
\includegraphics[width=3.5in]{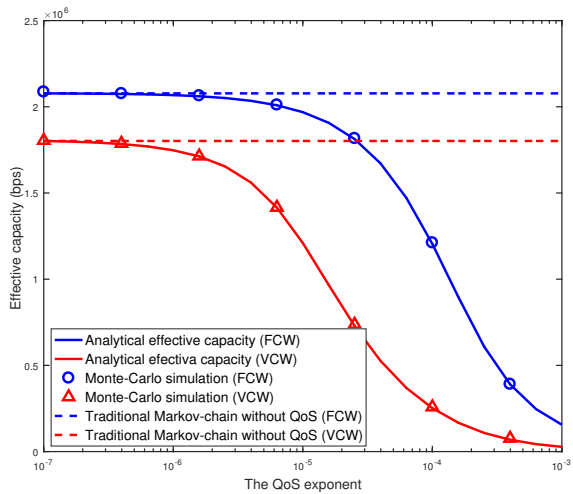}
\vspace{-0.4cm}
\caption{The effective capacity vs $\theta$, where $M=5$, $N=5$, and $K=1$.}
\label{figg1}
\end{figure}

In Fig. \ref{figg1}, the analytical effective capacity evaluated by (\ref{eq29}) is compared with the simulations results, where the numbers of WiFi nodes and LAA-BSs are $M = 5$ and $N = 5$, respectively, the number of LAA-UEs is $K=1$, and the bandwidth is $B=5$MHz. Each curve in the figure is plotted by increasing the persistent incoming rate of the LAA-BSs and evaluating the achieved QoS exponent $\theta$. The QoS exponent is specified by the $x$-axis, while the incoming rate is specified by the $y$-axis. We can see that the analysis, i.e., (\ref{eq29}), coincides the simulation results in both case of FCW and VCW; in other words, our analysis is accurate. We also see that the effective capacity is sensitive to $\theta$ within the region ${\rm{1}}{{\rm{0}}^{{\rm{ - 6}}}} \le \theta  \le {\rm{1}}{{\rm{0}}^{{\rm{ - 4}}}}$. In the case of $\theta>10^{-4}$, the QoS is too stringent and the effective capacity is small. In the case of $\theta<10^{-6}$, the QoS is too loose and the effective capacity is expected to approach to the capacity of the LAA-BS without QoS. In fact, we also analytically plot the capacity using the existing Markov chain analysis \cite{15,840210}, and confirm the convergence of the effective capacity and the capacity in the case without QoS. Further, we see that FCW can significantly outperform VCW in the presence of a small number of WiFi nodes and stringent QoS requirements. One reason is because VCW, though outperforming FCW in terms of the coexistence with WiFi (as discussed in Section III), can suffer from severe exponentially delayed retransmissions, violate QoS requirements, and therefore incur significant losses of effective capacity. Another reason is because FCW, sticking to a fixed small backoff window size, gives priority to LAA-BSs but can be more intrusive to WiFi, as compared to VCW, as discussed in Section I. However, FCW is less effective in terms of reacting to intensive collisions and therefore performs worse than VCW in the presence of large numbers of LAA-BSs and WiFi nodes, as will be shown later.

\begin{figure}[!t]
\centering
\includegraphics[width=3.5in]{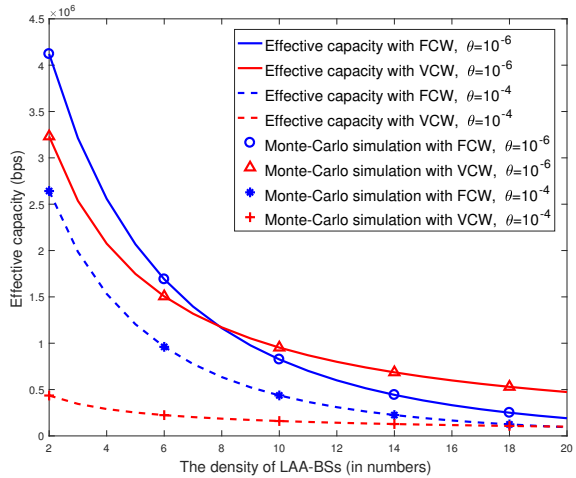}
\caption{The effective capacity of LAA versus $N$, where $M=5$, $K=1$, and $\theta=10^{-6},\,10^{-4}$.}
\label{figg2}
\end{figure}

\begin{figure}[!t]
\centering
\includegraphics[width=3.5in]{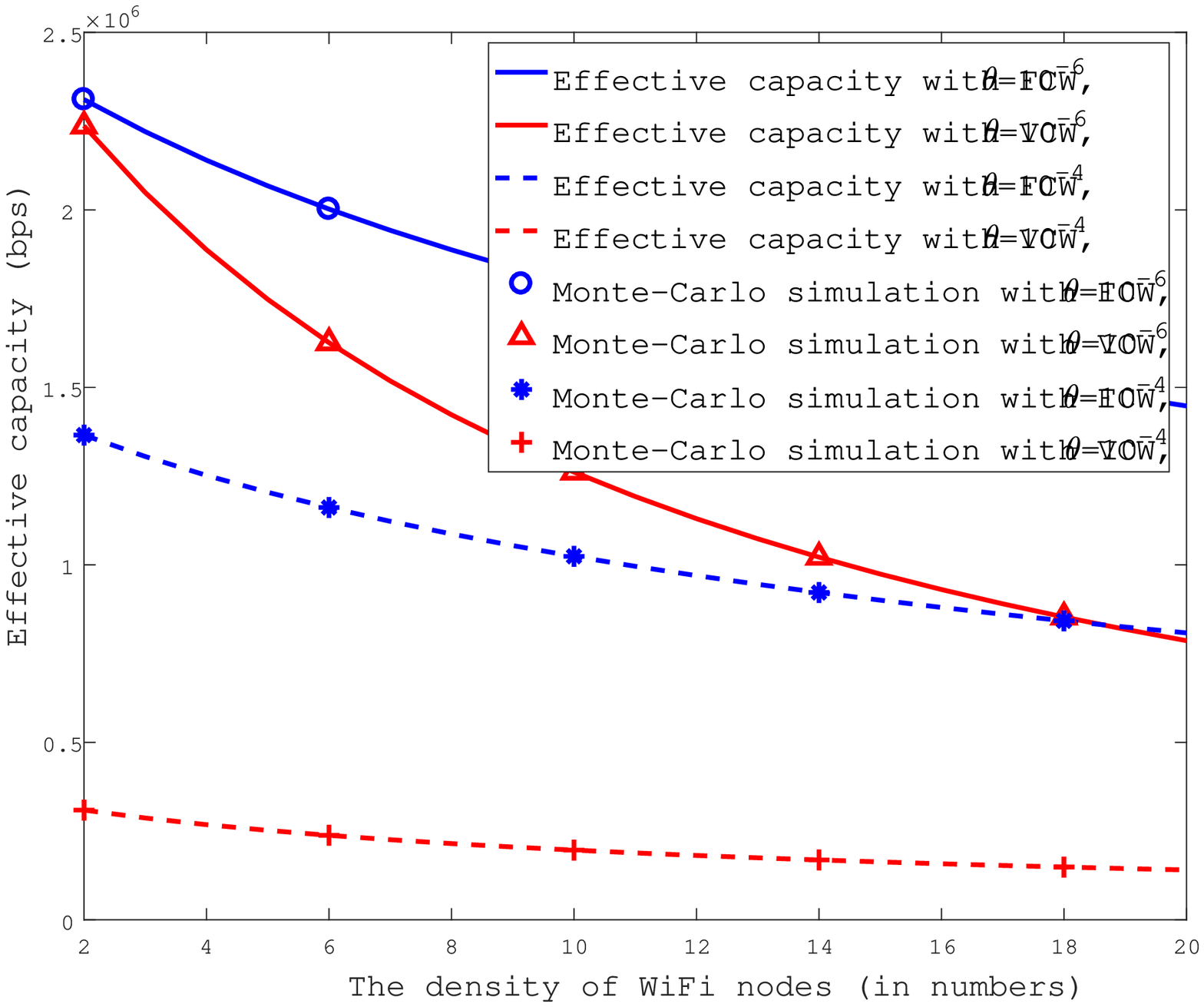}
\vspace{-0.4cm}
\caption{The effective capacity of LAA versus $M$, where $N = 5$, $K=1$, and $\theta=10^{-6},\,10^{-4}$.}
\label{figg3}
\vspace{-0.2cm}
\end{figure}

Fig. \ref{figg2} plots the effective capacity of LAA with the growing number of LAA-BSs $N$, where $M = 5$, $B=5$MHz, $K=1$, and $\theta=10^{-6},\,10^{-4}$. Validated by simulations, our analysis is once again confirmed to be accurate. The figure shows that the effective capacity of a LAA-BS decreases as $N$ increases. This is due to the increasing (re)transmission collisions. We also see that the decrease of the effective capacity slows down as $N$ grows. This is because the effective capacity is less susceptible to the number of LAA-BSs if there are more LAA-BSs and in turn the more intense collisions. Particularly, in the case that there are few LAA-BSs with mild collisions, FCW sticking to a small CW can get higher chances to access the channel over WiFi and consequently higher effective capacity. In the case that there are many LAA-BSs, VCW exponentially increasing the CW alleviates collisions and the loss of the effective capacity. In contrast, FCW suffers intense collisions and the effective capacity diminishes. In this sense, FCW is susceptible to the number of LAA-BSs, while VCW is robust.

Fig. \ref{figg3} plots the effective capacity of LAA with the increasing number of WiFi devices $M$, where $N = 5$, $B=5$MHz, $K=1$, and $\theta=10^{-6},\,10^{-4}$. We can see that VCW is more susceptible to the density of WiFi than FCW, given $N$ and $\theta$. As mentioned earlier, FCW allows a LAA-BS to stick to a small CW and as a result, gain priority over WiFi to access the channel. Therefore, FCW is less sensitive to the number of WiFi devices. In contrast, VCW exponentially increases the CW and increases the opportunities for WiFi devices to access the channel. Being more friendly to WiFi, VCW is more sensitive to the number of WiFi devices.

\begin{figure}[!t]
\centering
\includegraphics[width=3.5in]{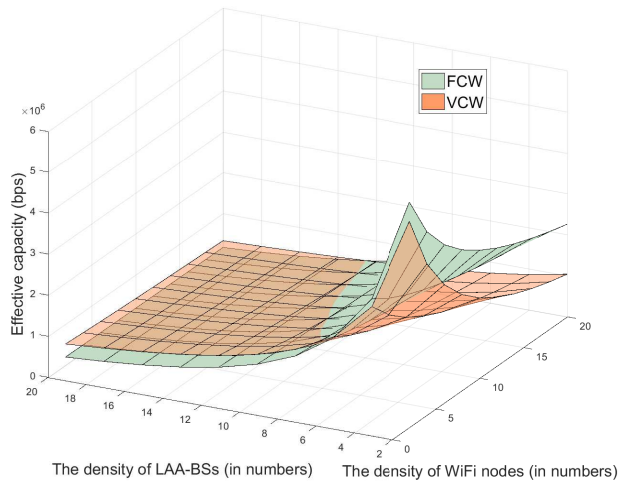}
\vspace{-0.4cm}
\caption{The effective capacity of LAA versus the respective densities of LAA-BSs and WiFi nodes, where $\theta=10^{-6}$.}
\label{figg4}
\vspace{-0.2cm}
\end{figure}

Extended from Figs. 5 and 6, Fig. \ref{figg4} provides a joint view of the effective capacity of LAA against both the densities of LAA-BSs and WiFi nodes, where $\theta ={10^{{\rm{ - }}6}}$. The conclusion drawn is that FCW is suitable for deployments with low density of LAA-BSs and high density of WiFi devices, while VCW is preferable for deployments with high density of LAA-BSs and low density of WiFi devices.

\begin{figure}[!t]
\centering
\includegraphics[width=3.5in]{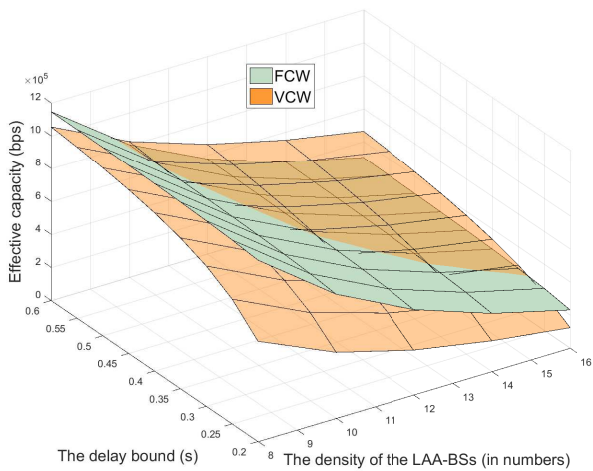}
\vspace{-0.4cm}
\caption{Effective capacity of LAA versus the density of LAA-BSs and traffic delay requirements.}
\label{figg5}
\vspace{-0.2cm}
\end{figure}

In Fig. \ref{figg5}, we evaluate the impact of the delay bound of traffic, $D_{\max}^k$, and the density of LAA-BSs, $N$, on the effective capacity of LAA, where $M = 5$, $B=5$MHz, and $K=1$. Given $D_{\max}^k$, the probability that the steady-state traffic delay at UE $k$ exceeds $D_{\max}^k$ is ${P^k_{th}} \approx  {\eta _k}{e^{ - {\theta_k}{C_k}({\theta_k }){D^k_{\max }}}}$ \cite{21}, where ${\eta _k}$ is the non-empty buffer probability, and can be approximated as the ratio of the constant arrival rate to the average transmit rate. Setting $P^k_{th}=0.1$, we can evaluate $C_k(\theta_k)$ by varying $D_{\max}^k$ and $N$. We can observe that the effective capacity grows with the delay bound in both cases of FCW and VCW. We also see that, FCW is more tolerant to the change of the delay bound, even when the bound is small. This is because FCW uses a small and fixed CW, and gains priority and earlier access the channel over WiFi nodes. In contrast, VCW is susceptible to small delay bounds, since it enlarges the CW to combat collisions at the cost of increased delays. On the other hand, VCW is more robust to the increasing number of LAA-BSs than FCW, as also observed in Figs. \ref{figg2} and \ref{figg4}. In this sense, a careful selection of FCW or VCW is important under different settings of delay bound and device numbers.

Fig. \ref{fig61} plots the maximum effective capacity of a LAA-BS by optimizing the transmit powers, as described in Section V-A.  Here, $M = 4$, $N = 1$, $K=20$, $B=20$MHz, and VCW is adopted. For comparison purpose, we also simulate a total channel inversion method \cite{Tang2007} and the water-filling method \cite{Goldsmith1997}. Water-filling only depends on the channel gains and maximizes the capacity. The total channel inversion method allocates the transmit power inversely proportionally to the channel gain of every LAA-UE, which is proved to be asymptotically optimal for maximizing the effective capacity of a single wireless point-to-point link as $\theta \rightarrow \infty$ \cite{Tang2007}. The figure shows that the proposed approach increasingly outperforms water-filling, as $\theta$ increases (i.e., the QoS becomes stringent). For instance, the gain of the proposed approach is up to 62.7\% in the case of $\theta=10^{-3}$. The proposed approach is indistinguishably close to water-filling, when $\theta\rightarrow 0$. This is because the effective capacity is equivalent to the capacity under loose QoS, while water-filling maximizes the capacity. On the other hand, the proposed approach outperforms the total channel inversion method across a wide range of $\theta\leq 10^{-2}$. For $\theta>10^{-2}$, the proposed approach provides the same performance as the total channel conversion method which is asymptotically optimal as $\theta\rightarrow \infty$.

\begin{figure}[!t]
\centering
\includegraphics[width=3.5in]{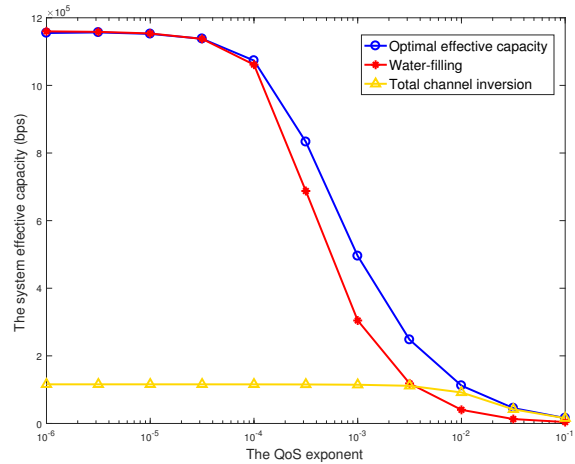}
\vspace{-0.4cm}
\caption{The effective capacity for different power allocation strategies.}
\label{fig61}
\vspace{-0.2cm}
\end{figure}

\begin{figure}[!t]
\centering
\includegraphics[width=3.5in]{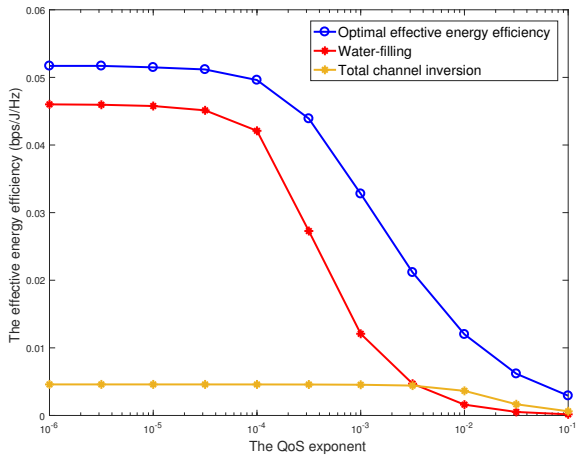}
\vspace{-0.4cm}
\caption{The effective energy efficiency for different power allocation strategies.}
\label{fig7}
\vspace{-0.2cm}
\end{figure}

Fig. \ref{fig7} plots the maximum effective energy efficiency of a LAA-BS by optimizing its transmit powers, as described in Section V-B. Here, $M = 4$, $N = 1$, $K=20$, $B=20$MHz, and VCW is adopted. $P_{\rm{static}}=P_{\rm{idle}}=0.1$ $W$, and $\xi=0.1$. The figure shows that the proposed approach is superior to the water-filling and total channel inversion in terms of effective energy efficiency. In the case of $\theta=10^{-3}$, the gains of the proposed approach are up to 171.4\% and 621.8\%, respectively.

It is worth pointing out that in the case that $\theta\rightarrow 0$ (e.g., $\theta\leq 10^{-5}$), the proposed approach provides a much higher effective energy efficiency than water-filling. This is different from the observation in Fig. \ref{fig61}. The reason is that the water-filling only can maximize the capacity, not energy efficiency, i.e., the ratio of capacity to the total power. When $\theta\rightarrow 0$, the effective energy efficiency recedes to the energy efficiency. Water-filling requires excessively high powers to maximize the capacity. In contrast, the proposed approach reduces the total power and leverages the power with the capacity, thereby maximizing the energy efficiency.

\section{Conclusion}
In this paper, we analyzed the effective capacity of LAA under statistically characterized QoS requirements and lossy wireless channel conditions. Closed-form expressions were derived to establish the connections between the effective capacity, QoS, channel conditions and transmission durations in a distributed heterogeneous network environment. The concavity of the effective capacity was revealed. Validated by simulations, the concavity was exploited to maximize the effective capacity and effective energy efficiency of LAA, and provided significant improvements of 62.7\%
and 171.4\%, respectively. Our analysis was of practical value to future holistic designs and deployments of LAA systems.

\appendices

\section{Evaluation of $\upsilon_L$ and $p_L$}
In the case of FCW, the transmission probability of a LAA-BS can be given by
\begin{equation}\label{eq2}
{\upsilon _L} = 1/{e_0} = \frac{2}{{{W_L} + 1}},
\end{equation}
where ${e_{\rm{0}}}$ is the mean backoff time.

In the case of VCW, the transmission probability of a LAA-BS can be given by
\begin{equation}\label{eq3}
{\upsilon _L} = \frac{{1 + {p_L} +  \cdots  + p_L^{{K_L} - 1}}}{{{e_0} + {p_L}{e_1} +  \cdots  + p_L^{{K_L} - 1}{e_{{K_L} - 1}}}},
\end{equation}
where ${e_j}$ is the mean backoff time of the $j$-th (re)transmission for a packet in LAA.

Given collision probability ${p_W}$, the transmission probability of a WiFi node can be given by \cite{22}
\begin{equation}\label{eq1}
{\upsilon _W} = \frac{{1 + {p_W} +  \cdots  + p_W^{{K_W} - 1}}}{{{b_0} + {p_W}{b_1} +  \cdots  + p_W^{K{}_W - 1}{b_{{K_W} - 1}}}},
\end{equation}
where ${b_j}$ is the mean backoff time of the $j$-th retransmission for a packet in WiFi, and $K_W$ is the maximum number of retransmissions of WiFi per packet.

The collision probabilities of a WiFi node and a LAA-BS can be obtained by solving
\begin{equation}\label{eq4}
{p_W} = 1 - {(1 - {\upsilon _W})^{M - 1}}{(1 - {\upsilon _L})^N},
\end{equation}
\begin{equation}\label{eq5}
{p_L} = 1 - {(1 - {\upsilon _W})^M}{(1 - {\upsilon _L})^{N - 1}}.
\end{equation}
By Brouwer's fixed point theorem \cite{22}, there exists a fixed point or unique solution for (\ref{eq4}) and (\ref{eq5}). $\upsilon_L$, $p_L$, $\upsilon_W$ and $p_W$ are readily available.

Our model can be applied in the presence of hidden nodes. Particularly, hidden nodes would affect $p_L$, $p_W$ and $f(\tau)$, but would not affect our model where $p_L$, $p_W$ and $f(\tau)$ are just inputs. Extended from ~\cite{Ekici2008,Jang2012}, $p_L$ and $p_W$ can be rewritten in the presence of hidden nodes, as given by
\begin{equation*}
\begin{aligned}
{p_L}\!=&1\!-\!{(1\! - \!{v_L})^{{N} - 1}}{(1 \!-\! {v_W})^{{M}}}{\left[ {{{(1 - {v_L})}^{{h_L}}}{{(1 - {v_W})}^{{h_W}}}} \right]^{\frac{{2{T_s}}}{{\bar \tau }}}}\\
{p_W}\!=&1\!-\!{(1 \!- \!{v_L})^{{N} }}{(1 \!- \!{v_W})^{{M}-1}}{\left[ {{{(1 - {v_L})}^{{h_L}}}{{(1 - {v_W})}^{{h_W}}}} \right]^{\frac{{2{T_s}}}{{\bar \tau }}}}
\end{aligned}
\end{equation*}
where $h_L$ and $h_W$ are the numbers of hidden LAA base stations and hidden WiFi nodes, respectively, and $f(\tau)$ and $\bar{\tau}$ can be explicitly written in terms of $p_W$, $v_W$, $p_L$ and $v_L$.
\newcounter{mytempeqncnt1}
\begin{figure*}%
\normalsize
\begin{equation}
\label{eq1022}
{f}(\tau) = \left\{ \begin{array}{l}
\Pr\{ {\tau} = {\sigma _{\rm {idle}}}\}  = {(1 - {\upsilon _L})^{N - 1}}{(1 - {\upsilon _W})^M}\\
\Pr \{ \tau  = {T_{{\rm{cw}}}}\} = (1 - {(1 - {\upsilon _W})^M} - M{\upsilon _W}{(1 - {\upsilon _W})^{M - 1}}){(1 - {\upsilon _L})^{N - 1}}\\
\Pr\{ {\tau} = {T_{\rm c}}\}  = {(1 - {\upsilon _W})^M}(1 - {(1 - {\upsilon _L})^{N - 1}} - (N - 1){\upsilon _L}{(1 - {\upsilon _L})^{N - 2}})\\
\Pr\{ {\tau} = {T_{\rm {sw}}}\}  = {(1 - {\upsilon _L})^{N - 1}}M{\upsilon _W}{(1 - {\upsilon _W})^{M - 1}}\\
\Pr\{ {\tau} = {T_{\rm f}}\}  = {(1 - {\upsilon _W})^M}(N - 1){\upsilon _L}{(1 - {\upsilon _L})^{N - 2}}\\
\Pr\{ {\tau} = {T_{\rm {wl}}}\}  = 1-\Pr\{ {\tau} = {\sigma _{\rm {idle}}}\}  - \Pr\{ {\tau} = {T_{\rm {cw}}}\}  - \Pr\{ {\tau} = {T_{\rm {c}}}\}  - \Pr\{ {\tau} = {T_{\rm {sw}}}\}  - \Pr\{ {\tau} = {T_{\rm {f}}}\}
\end{array} \right.
\end{equation}
\hrulefill %
\vspace*{4pt}
\end{figure*}

\section{Calculation of the PGF}

Let ${\tau}$ denote the duration of a timeslot. It can take from six values: $\sigma _{\rm {idle}}$, ${T_{\rm {cw}}}$, ${T_{\rm c}}$, ${T_{\rm {sw}}}$, ${T_{\rm f}}$ and ${T_{\rm {wl}}}$. Here, $\delta_{\rm {idle}}$ corresponds to an idle slot; $T_{\rm {cw}}$ corresponds to a slot with collisions between WiFi nodes; $T_{\rm c}$ corresponds to a slot with collisions between the other LAA-BSs; $T_{\rm {sw}}$ corresponds to a slot with a successful WiFi transmission; $T_{\rm f}$ corresponds to a slot with a successful transmission of other LAA-BSs; and $T_{\rm {wl}}$ corresponds to a slot with collisions between WiFi nodes and other LAA-BSs.

The probability mass function (PMF) of ${\tau}$, denoted by $f(\tau)$, is given in (\ref{eq1022}), where $p_L$ and $\upsilon_L$ are the collision probability and the transmission probability of a LAA-BS per timeslot, respectively; and $p_W$ and $\upsilon_W$ are the collision probability and the transmission probability of a WiFi node per timeslot, respectively. These parameters can be calculated in Appendix A. In (\ref{eq1022}), we assume that ${T_{\rm c}} = {T_{\rm f}} = {T_{\rm {wl}}}$. $T_{\rm {sw}}$ and $T_{\rm {cw}}$ can be different in the different WiFi modes. The PGF and the mean of $\tau$ are given respectively by
\begin{equation}\label{eq188}
\begin{aligned}
\hat \tau(z) &= \Pr\{ {\tau} = {\sigma _{\rm {idle}}}\} {z^{{\sigma _{\rm {idle}}}}} + \Pr\{ {\tau} = {T_{\rm {cw}}}\} {z^{{T_{\rm {cw}}}}} \\
&+ \Pr\{ {\tau} = {T_{\rm {c}}}\} {z^{{T_{\rm {c}}}}} + \Pr\{ {\tau} = {T_{\rm {sw}}}\} {z^{{T_{\rm {sw}}}}} \\
&+ \Pr\{ {\tau} = {T_{\rm {f}}}\} {z^{{T_{\rm {f}}}}} + \Pr\{ {\tau} = {T_{\rm {wl}}}\} {z^{{T_{\rm {wl}}}}},
\end{aligned}
\end{equation}
\begin{equation}\label{eq18}
\begin{aligned}
\bar \tau  =& \Pr \{ \tau  = {\sigma _{{\rm{idle}}}}\} {\sigma _{{\rm{idle}}}} + \Pr \{ \tau  = {T_{{\rm{cw}}}}\} {T_{{\rm{cw}}}} + \Pr \{ \tau  = {T_{\rm{c}}}\} {T_{\rm{c}}}\\
&+ \Pr \{ \tau  = {T_{{\rm{sw}}}}\} {T_{{\rm{sw}}}} + \Pr \{ \tau  = {T_{\rm{f}}}\} {T_{\rm{f}}} + \Pr \{ \tau  = {T_{{\rm{wl}}}}\} {T_{{\rm{wl}}}}.
\end{aligned}
\end{equation}

\end{document}